\documentclass[preprint,12pt]{elsarticle}

\usepackage[T1]{fontenc}
\usepackage[utf8]{inputenc}
\usepackage{amsmath,amssymb,amsfonts,amsthm}
\usepackage{bm}
\usepackage{booktabs,array}
\usepackage{graphicx}
\usepackage{xurl}
\usepackage{placeins}
\usepackage{float}
\usepackage{algorithm}
\usepackage{algorithmic}
\usepackage{tikz}
\usetikzlibrary{arrows.meta,positioning,shapes.geometric}
\usepackage[colorlinks=true,allcolors=blue]{hyperref}
\biboptions{sort&compress}

\graphicspath{{figures/}}

\newtheorem{definition}{Definition}
\newtheorem{assumption}{Assumption}
\newtheorem{theorem}{Theorem}
\newtheorem{lemma}{Lemma}
\newtheorem{proposition}{Proposition}
\newtheorem{corollary}{Corollary}
\newtheorem{remark}{Remark}
\newcommand{\R}{\mathbb{R}}
\newcommand{\K}{\mathcal K}
\newcommand{\KL}{\mathcal{KL}}
\newcommand{\norm}[1]{\left\lVert #1\right\rVert}

\newcommand{\col}{\operatorname{col}}

\journal{}

\begin{document}

\begin{frontmatter}

\title{Input-to-State Stability Certification via Projection Residuals for Koopman Learning Control of Nonlinear Repetitive Systems}

\author[xjtu,tobacco]{Yue Wu\corref{cor1}}
\ead{wuyue0619@stu.xjtu.edu.cn}
\author[xjtu]{Ye Cao}
\author[xjtu]{Jianfu Cao}

\cortext[cor1]{Corresponding author.}

\affiliation[xjtu]{organization={School of Automation, Xi'an Jiaotong University},
            city={Xi'an},
            postcode={710049},
            country={China}}

\affiliation[tobacco]{organization={Xinjiang Cigarette Factory, Hongyun Honghe Tobacco (Group) Co., Ltd.},
            city={Urumqi},
            postcode={830000},
            country={China}}

\begin{abstract}
This paper studies input-to-state stability (ISS) certification for
data-driven Koopman learning control of unknown discrete-time nonlinear
repetitive systems over finite trial horizons. Rather than proposing a new
learning law, we certify when a fixed Koopman-assisted constrained update
yields a practical stability bound for the selected tracking error along the
trial axis. Prediction accuracy alone is insufficient for this purpose: the
selected finite-horizon input-output channel must have a positive margin,
and the unreachable component of the requested output increment must be
accounted for through a projection residual. Thus, a Koopman predictor with
small held-out prediction residuals may still fail the learning-stability
certificate if its selected channel is weak. We formulate the selected
stacked tracking error as the state of a discrete-time learning-axis system
and treat Koopman residuals, reset mismatch, channel uncertainty, projection
residuals, deployment shifts, and numerical tolerances as ISS inputs. The
deterministic result gives a practical ISS estimate from the initial
learning error to an explicit ultimate band. A finite-sample implementation
constructs an episode-level residual bound under a fixed controller and
combines it with reported channel, projection, shift, and numerical margins.
Numerical checks on nonlinear repetitive systems support the predicted
residual-to-band scaling, weak-channel rejection, projection closure, and
ultimate-band coverage.
\end{abstract}

\begin{highlights}
\item Develops a trial-axis input-to-state stability certificate for Koopman learning control.
\item Separates prediction residuals from selected-channel learning certifiability.
\item Quantifies the projection-residual contribution to the ultimate error band.
\item Combines episode-level residual calibration with channel, projection, and shift margins.
\item Evaluates residual scaling, weak-channel rejection, projection closure, and band coverage.
\end{highlights}


\end{frontmatter}

\section{Introduction}

Stability is a central issue in learning control. In a
repetitive task, the controller uses data from previous trials to update the
input for future trials; the same mechanism that improves tracking can also
amplify modelling errors, nonrepetitive disturbances, reset mismatch,
sensor noise, and actuator constraints. Classical iterative learning control
(ILC) therefore developed lifted-domain, two-dimensional repetitive-system,
norm-optimal, robust, and monotonic-convergence analyses that clarify when
tracking errors converge, when learning transients remain bounded, and how
uncertainty affects final tracking performance
~\cite{arimoto1984,moore1993,amann1996,gao2001robust,bristow2006survey}.

The stability question addressed here is different from classical zero-error
convergence or nominal monotonic convergence. In data-driven Koopman
learning control, the finite-dimensional predictor is learned from data and
is used through a selected finite-horizon input-output channel. A small
prediction residual does not imply that this selected channel can generate
the requested output correction under input-update constraints. Moreover,
reset mismatch, nonrepetitive disturbances, channel-identification errors,
deployment shifts, and numerical tolerances do not disappear across trials.
They produce a nonzero learning-axis perturbation budget and therefore lead,
at best, to a certified ultimate band rather than exact convergence to zero.
This separates model fit from learning certifiability: the former concerns
prediction residuals, whereas the latter also requires a selected-channel
margin and request-set reachability.

Input-to-state stability (ISS) is used here because it gives the certificate
needed for learning under persistent perturbations. The effect of the
initial learning error decays through a \(\mathcal{KL}\) term, while
residuals, reset mismatch, channel uncertainty, projection errors, and
deployment shifts enter through explicit gains. Input-output stability of a
finite-horizon plant channel can bound output responses to input
trajectories, but it does not by itself provide a trial-axis decay estimate
from the initial learning error to a computable ultimate band. Internal
stability of the physical plant concerns the within-trial time variable
\(t\), whereas the present certificate concerns the learning recursion
indexed by the trial number \(k\). Since the original ISS and small-gain
developments
~\cite{sontag1989,sontagwang1995,jiangteelpraly1994,jiangwang2001},
this viewpoint has been extended to integral, incremental, large-scale, and
infinite-dimensional settings
~\cite{angelisontagwang2000,angeli2002,dashkovskiy2010,mironchenko2018}.

Koopman operator theory and data-driven finite-dimensional approximations
provide useful predictors for nonlinear systems
~\cite{koopman1931,rowley2009,williams2015,brunton2016,proctor2016,korda2018,mauroy2020}.
Data-driven control and finite-data Koopman analyses likewise show that
uncertainty margins must be part of controller certification, not only part
of model assessment~\cite{nueske2023,dean2020,depersis2020,berberich2021}.
The technical focus is the construction of the aggregate ultimate-band
constant from quantities available before deployment: complete-episode
residual calibration, finite-horizon channel margins, projection residuals,
reset mismatch, deployment shift, and numerical tolerance.

The contributions are as follows. First, we formulate the stability of
Koopman learning control for unknown nonlinear repetitive systems as a
practical ISS problem along the learning-trial axis. The selected stacked
tracking error is the state, while Koopman residuals, reset mismatch,
channel uncertainty, projection residuals, deployment shifts, and numerical
tolerances are treated as input channels. Second, we identify projection
residuals and selected-channel margins as necessary certificate objects for
Koopman-assisted learning control. This shows that prediction accuracy alone
is not sufficient: a learned predictor with small residual can still be
non-certifiable if the selected finite-horizon input-output channel is weak
or if the requested output increment is outside the constrained reachable
set. Third, we provide a computable certificate budget that combines
complete-episode residual calibration, finite-horizon channel and reset
margins, projection residuals, and implementation margins. The resulting
practical ISS estimate gives an explicit certified ultimate band for the
learning error. Fourth, numerical studies audit the certificate components
predicted by the theory, including residual-to-band scaling,
selected-channel rejection, projection closure, and ultimate-band coverage.

\section{System Class and Learning-Axis Stability Formulation}

We consider a nonlinear repetitive system operated over a finite within-trial
horizon \(t=0,\ldots,N-1\) and a learning-trial index \(k=0,1,\ldots\):
\begin{equation}
x_k(t+1)=f(x_k(t),u_k(t),d_k(t)),\qquad
y_k(t)=h(x_k(t)).
\label{eq:plant}
\end{equation}
The functions \(f\) and \(h\) are unknown. The reference \(y_r(t)\) is
repeated, but the reset \(x_k(0)\), disturbance \(d_k\), sensor noise, and
actuator mismatch may vary from trial to trial.

\begin{assumption}[Certified nonlinear repetitive task class]
We consider discrete-time nonlinear systems operated repeatedly over a
finite horizon \(t=0,\ldots,N-1\) and indexed by the trial number \(k\).
The maps \(f\) and \(h\) are unknown, but the task protocol is resettable
and the reference trajectory is repeated. The state, input, output,
disturbance, reset mismatch, and actuator update are restricted to certified
bounded sets on which solutions are well defined and the chosen lifting map
\(\psi\) is evaluable. The within-trial plant is not assumed to be globally
ISS. The certificate developed below concerns the selected stacked tracking
error along the learning-trial axis under the declared finite-horizon
protocol.
\end{assumption}

Thus, the method applies to finite-horizon nonlinear repetitive tasks, such
as robot path repetition, batch operation, mechatronic servo tasks, and
other resettable processes for which input-output trajectories can be
collected over repeated trials and a frozen Koopman finite-horizon channel
can be calibrated before deployment.

Let
\(E_k\in\R^{Np}\) denote the stacked tracking error and let \(S_H\) select
the rows used for certification. The controlled trial-axis state is
\begin{equation}
\bar E_k=S_HE_k-S_HE_k^\star,
\label{eq:selected-error}
\end{equation}
where \(E_k^\star\) is a prescribed transient error profile used to absorb
bounded initial mismatch during the initial segment of each trial.

The variable \(\bar E_k\) is a learning-axis state rather than a physical
plant state. Exact or monotone convergence to zero would require ideal
resets, exact models, and feasible inverse updates. Under Koopman residuals,
nonrepetitive disturbances, finite-data channel errors, actuator constraints,
and numerical tolerances, the relevant question is instead whether
\(\{\bar E_k\}\) enters a computable ultimate band as the trial index \(k\)
increases. Input-output stability of the finite-horizon plant channel bounds
output responses to input trajectories, whereas internal stability concerns
the within-trial dynamics \(x_k(t)\). We therefore use practical ISS on the
trial axis. The terms learning axis and trial axis both refer to the discrete
index \(k\) of repeated task executions.

\begin{definition}[Trial-axis practical ISS]
The selected learning-error system is practically ISS if there exist
\(\beta\in\mathcal{KL}\), \(\gamma\in\mathcal{K}\), and \(c\geq0\) such that, for every
initial selected error and every admissible perturbation sequence \(v\),
\begin{equation}
\norm{\bar E_k}_\infty
\leq
\beta(\norm{\bar E_0}_\infty,k)
+\gamma\!\left(\sup_{0\leq j<k}\norm{v_j}\right)+c
\label{eq:piss-def}
\end{equation}
for all certified trials \(k\). If \(c=0\), the system is ISS on the
certified deployment grid.
\end{definition}

The norm \(\norm{v_k}\) in \eqref{eq:piss-def} is any fixed product norm
compatible with the componentwise bounds used below. In the deterministic
certificate we use the scalar budget norm
\begin{align*}
\norm{v_k}_{\rm ISS}={}&
\norm{\omega_k}_\infty+\norm{p_k}_\infty
+\epsilon_G\norm{\Delta U_k}_\infty\\
&+\epsilon_O\norm{\Delta z_{k,0}}_\infty
+h_{{\rm shift},k}+\norm{\omega_{{\rm num},k}}_\infty .
\end{align*}

The perturbation \(v_k\) collects the quantities that can be certified or
calibrated in the construction below:
\begin{equation}
v_k=\col(\omega_k,p_k,\Delta G_k\Delta U_k,\Delta O_k\Delta z_{k,0},
h_{{\rm shift},k},\omega_{{\rm num},k}).
\label{eq:iss-input}
\end{equation}
Here \(\omega_k\) is the selected model residual, \(p_k\) is the
projection residual of the constrained input update, \(\Delta G_k\) and
\(\Delta O_k\) are learned-channel and observability perturbations,
\(\Delta z_{k,0}\) is the lifted reset difference, and the last two terms
cover deployment-shift and numerical effects. Table~\ref{tab:iss-inputs}
summarizes the role of each component.

\begin{table}[t]
\caption{Input-to-State Stability Inputs in the Trial-Axis Learning Error System}
\label{tab:iss-inputs}
\centering
\scriptsize
\setlength{\tabcolsep}{2.5pt}
\begin{tabular}{p{0.20\linewidth}p{0.29\linewidth}p{0.33\linewidth}}
\toprule
Input component & Source & Certificate role\\
\midrule
\(\omega_k\) & Koopman and output residual, nonrepetitive disturbance & Additive selected-output perturbation\\
\(p_k\) & Constrained or regularized input solve & Unreachable requested output increment\\
\(\Delta G_k\Delta U_k\) & Learned-channel error & Gain from channel perturbation to output error\\
\(\Delta O_k\Delta z_{k,0}\) & Reset/observability perturbation & Initial-shift contribution\\
\(h_{{\rm shift},k}\) & Calibration-to-deployment mismatch & Distribution-shift margin\\
\(\omega_{{\rm num},k}\) & Solver and truncation tolerance & Numerical robustness margin\\
\bottomrule
\end{tabular}
\end{table}

\section{Koopman Finite-Horizon Channel and Projection Residual}

Let \(z=\psi(x)\in\R^q\) be a vector of observables. A frozen extended
dynamic mode decomposition with control (EDMDc) model is fitted before
deployment,
\begin{equation}
z(t+1)=\widehat A z(t)+\widehat B u(t)+\xi(t),\qquad
y(t)=\widehat C z(t)+\nu(t).
\label{eq:edmdc}
\end{equation}
Unrolling \eqref{eq:edmdc} over a finite horizon gives a stacked
input-output relation. After selecting the certified rows,
\begin{equation}
\Delta \bar Y_k
=
O_{\star,H}\Delta z_{k,0}
+G_{\star,H}\Delta U_k+\omega_k ,
\label{eq:true-channel}
\end{equation}
where \(\Delta U_k=U_k-U_{k-1}\), \(\Delta \bar Y_k=S_H(Y_k-Y_{k-1})\),
and \(G_{\star,H}\) is the true selected finite-horizon
Koopman--Markov channel. The controller uses frozen learned quantities
\(\hat O_H\) and \(\hat G_H\).

The constrained update is used as a fixed Koopman-assisted learning step;
the contribution lies in the associated stability certificate, namely the
selected-channel margin, projection residual, and computable input-to-state
stability budget.

For a diagonal learning matrix \(\Lambda\) satisfying
\(\norm{\Lambda}_\infty\leq\lambda<1\), the requested selected output
increment is
\begin{equation}
b_k=S_H(E_{k-1}^\star-E_k^\star)
-(I-\Lambda)\bar E_{k-1}
+\hat O_H\Delta z_{k,0}.
\label{eq:request}
\end{equation}
The admissible input update is obtained from the constrained least-squares
problem
\begin{equation}
\Delta U_k\in
\arg\min_{\Delta U\in{\mathcal U}_\Delta}
\norm{\hat G_H\Delta U-b_k}_2^2+\mu\norm{\Delta U}_2^2 ,
\label{eq:controller}
\end{equation}
and the residual
\begin{equation}
p_k=b_k-\hat G_H\Delta U_k
\label{eq:projection-residual}
\end{equation}
is retained in the stability proof. Unlike a nominal learned inverse, an
increment outside the reachable set \(\hat G_H{\mathcal U}_\Delta\)
contributes its unreachable component directly as an ISS input.

\begin{proposition}[Projection-residual obstruction]
Let \(G=\hat G_H\), let \({\mathcal B}\) be a certified request set
containing all \(b_k\), and let \({\mathcal U}_\Delta\) be the admissible
input-update set. A uniform zero-projection learning step over
\({\mathcal B}\) is possible only if
\begin{equation}
{\mathcal B}\subseteq G{\mathcal U}_\Delta .
\label{eq:request-inclusion}
\end{equation}
More generally, any uniform ISS budget over \({\mathcal B}\) must contain
at least
\begin{equation}
\bar p_{\mathcal B}
=
\sup_{b\in{\mathcal B}}
\inf_{\Delta U\in{\mathcal U}_\Delta}
\norm{b-G\Delta U}_\infty .
\label{eq:projection-lower-bound}
\end{equation}
The projection residual thus equals the distance from the requested output
increment to the constrained reachable set of the learned channel,
reflecting a structural limitation on actuation rather than numerical error.
\end{proposition}

\begin{proof}
By definition, \(p_k=b_k-G\Delta U_k\). If some
\(b\in{\mathcal B}\) is not in \(G{\mathcal U}_\Delta\), no admissible
\(\Delta U\) can make \(p_k=0\) for that request. In the nominal case with
no Koopman residual, no reset perturbation, and no channel perturbation, the
selected recursion contains the additive term \(p_k\). Therefore a uniform
ISS budget that applies to all \(b\in{\mathcal B}\) must be at least the
worst-case distance in \eqref{eq:projection-lower-bound}.
\end{proof}

\begin{corollary}[Projection-induced ultimate-band contribution]
Consider any uniform certificate of the form
\[
\norm{\bar E_k}_\infty
\leq \lambda\norm{\bar E_{k-1}}_\infty+\bar w,\qquad
0\leq\lambda<1,
\]
over a request set \({\mathcal B}\) and an admissible update set
\({\mathcal U}_\Delta\). If
\[
\bar p_{\mathcal B}
=
\sup_{b\in{\mathcal B}}
\inf_{\Delta U\in{\mathcal U}_\Delta}
\norm{b-G\Delta U}_\infty ,
\]
then the certified input budget must satisfy
\(\bar w\geq\bar p_{\mathcal B}\). Consequently, within any uniform
worst-case ultimate-band certificate over \({\mathcal B}\) based on this
recursion, the projection obstruction contributes at least
\(\bar p_{\mathcal B}/(1-\lambda)\).
\end{corollary}

\begin{proof}
Proposition~1 gives \(\bar w\geq\bar p_{\mathcal B}\) for any uniform
budget over \({\mathcal B}\). Substituting this lower bound into the
steady-state term \(\bar w/(1-\lambda)\) yields the stated
projection-induced contribution.
\end{proof}

\begin{proposition}[Projection budget and request closure]
Let \(G=\hat G_H\) and let \({\mathcal B}\) be a certified request set
containing all \(b_k\). For the regularized controller
\eqref{eq:controller}, a computable projection budget is
\begin{equation}
\bar p_{\mu,{\mathcal U}}
=\sup_{b\in{\mathcal B}}
\min_{\Delta U\in{\mathcal U}_\Delta}
\left(\norm{G\Delta U-b}_2^2+\mu\norm{\Delta U}_2^2\right)^{1/2}.
\label{eq:p-budget}
\end{equation}
Then \(\norm{p_k}_\infty\leq \bar p_{\mu,{\mathcal U}}\). If
\({\mathcal U}_\Delta=\R^{Nm}\), \(\mu=0\), and \(G\) has full row rank
with \(\sigma_{\min}(G)\geq\gamma_G>0\), then \(p_k=0\) and
\(\norm{\Delta U_k}_2\leq\norm{b_k}_2/\gamma_G\).
\end{proposition}

\begin{proof}
For every certified request \(b_k\in{\mathcal B}\), the optimizer of
\eqref{eq:controller} attains an objective value no larger than the
right-hand side of \eqref{eq:p-budget}. Since
\(\norm{p_k}_\infty\leq\norm{p_k}_2\), the claimed projection bound follows.
If \({\mathcal U}_\Delta=\R^{Nm}\), \(\mu=0\), and \(G\) has full row rank,
then \(G\Delta U=b_k\) is feasible for every selected request. The
minimum-norm solution satisfies
\(\Delta U_k=G^\dagger b_k\), hence \(p_k=0\) and
\(\norm{\Delta U_k}_2\leq\norm{G^\dagger}_2\norm{b_k}_2
\leq\norm{b_k}_2/\gamma_G\).
\end{proof}

\begin{proposition}[Weak-channel non-certifiability]
Suppose the learned selected channel and its spectral perturbation radius
satisfy
\begin{equation}
\sigma_{\min}(\hat G_H)-\epsilon_{G,2}\leq 0 .
\label{eq:weak-channel-test}
\end{equation}
Then the data cannot certify a positive lower input-output stability (IOS)
margin for the true selected channel \(G_{\star,H}\). Consequently, no uniform
inverse-gain or projection-closure certificate of the form used below can be
issued from these data.
\end{proposition}

\begin{proof}
The singular-value perturbation inequality gives
\[
\sigma_{\min}(G_{\star,H})
\geq
\sigma_{\min}(\hat G_H)-\norm{\hat G_H-G_{\star,H}}_2 .
\]
If the certified perturbation radius is \(\epsilon_{G,2}\) and
\eqref{eq:weak-channel-test} holds, the lower bound supplied by the data is
nonpositive. Hence the data do not certify a positive channel margin
\(\gamma_G>0\), which is required for the inverse-gain and request-closure
steps.
\end{proof}

\begin{definition}[Deterministic request set]
Let \({\mathbb B}_\infty(r)\) denote an infinity-norm ball of the dimension
implied by context.
For a certified error radius \(R\) and reset-lift radius \(\bar Z_0\), define
\begin{align}
{\mathcal B}(R,\bar Z_0)
={}&{\mathcal B}_\star
\oplus (I-\Lambda){\mathbb B}_\infty(R)
\oplus \hat O_H{\mathbb B}_\infty(\bar Z_0),
\label{eq:det-request-set}\\
{\mathcal B}_\star
={}&\operatorname{co}\{S_H(E_{k-1}^\star-E_k^\star):k\geq1\}.
\nonumber
\end{align}
\end{definition}

\begin{lemma}[Request-set containment]
If \(\norm{\bar E_{k-1}}_\infty\leq R\) and
\(\norm{\Delta z_{k,0}}_\infty\leq\bar Z_0\), then the request
\eqref{eq:request} belongs to \({\mathcal B}(R,\bar Z_0)\).
\end{lemma}

\begin{proof}
The three terms in \eqref{eq:request} belong respectively to
\({\mathcal B}_\star\), \((I-\Lambda){\mathbb B}_\infty(R)\), and
\(\hat O_H{\mathbb B}_\infty(\bar Z_0)\). Their sum is therefore in the
Minkowski sum \eqref{eq:det-request-set}.
\end{proof}

\begin{proposition}[Certified request closure]
Let \(P(R,\bar Z_0)\) be the worst projection budget
\eqref{eq:p-budget} over \({\mathcal B}(R,\bar Z_0)\), and let
\(\bar\xi\) collect all non-projection perturbation terms. The certified
ball \(\{\bar E:\norm{\bar E}_\infty\leq R\}\) is invariant for the
one-step learning recursion if
\begin{equation}
\lambda R+P(R,\bar Z_0)+\bar\xi\leq R.
\label{eq:request-closure}
\end{equation}
\end{proposition}

\begin{proof}
Substituting the update request into the selected learning recursion leaves
the contracted term \(\Lambda\bar E_{k-1}\), the projection residual, and
the remaining perturbation terms. Hence every state satisfying
\(\norm{\bar E_{k-1}}_\infty\leq R\) obeys
\(\norm{\bar E_k}_\infty\leq\lambda R+P(R,\bar Z_0)+\bar\xi\). The
right-hand side is at most \(R\) under \eqref{eq:request-closure}.
\end{proof}

Thus \eqref{eq:request-closure} is a request-closure, or certified
invariance, condition. It is not a small-gain theorem for an interconnection;
it checks that the output increments generated by the ISS state remain
inside the certified finite-horizon reachable set.

\begin{proposition}[Support-function projection audit]
If \({\mathcal B}(R,\bar Z_0)\) and
\({\mathcal U}_\Delta\) are compact convex sets, then the zero-projection
condition \({\mathcal B}(R,\bar Z_0)\subseteq G{\mathcal U}_\Delta\) is
equivalent to
\[
h_{{\mathcal B}(R,\bar Z_0)}(s)
\leq h_{{\mathcal U}_\Delta}(G^\top s)
\quad\text{for all }s .
\]
For a finite audited direction set \({\mathcal S}\), the reported
projection margin can be lower bounded by the maximum positive violation
\[
\max_{s\in{\mathcal S},\,\norm{s}_1\leq1}
\left[h_{{\mathcal B}(R,\bar Z_0)}(s)
-h_{{\mathcal U}_\Delta}(G^\top s)\right]_+ .
\]
\end{proposition}

\begin{proof}
The first statement is the standard support-function characterization of
convex-set inclusion. The second follows by restricting the separating
directions to the audited finite set and using the dual representation of
the infinity-norm distance.
\end{proof}

\begin{algorithm}[H]
\caption{Trial-Axis ISS Certification Protocol}
\label{alg:iss-certification}
\begin{algorithmic}[1]
\REQUIRE Data splits \({\mathcal D}_\Theta,{\mathcal D}_C,
{\mathcal D}_{\rm des},{\mathcal D}_{\rm cal}\); horizon \(H\);
candidate set \({\mathcal Q}\) of \((S_H,\Lambda,\mu,{\mathcal U}_\Delta)\);
radius \(R\), target band \(\Delta_{\rm req}\), risk levels.
\ENSURE Status \(\in\{\textsc{Certified},\textsc{NotCertified}\}\),
frozen controller, budget \(\bar w_\beta\), band
\(\Delta_{{\rm ISS},\beta}\).
\STATE Fit \((\hat A,\hat B)\) from \({\mathcal D}_\Theta\) and fit
\(\hat C\) from \({\mathcal D}_C\).
\FOR{each candidate \(q=(S_H,\Lambda,\mu,{\mathcal U}_\Delta)\in{\mathcal Q}\)}
\STATE Compute \(\hat O_H(q)\), \(\hat G_H(q)\), and
\(\lambda=\norm{\Lambda}_\infty\).
\IF{\(\lambda\geq 1\)}
    \STATE discard \(q\); \textbf{continue}
\ENDIF
\STATE Build the deterministic request set
\({\mathcal B}(R,\bar Z_0)\) from \eqref{eq:det-request-set}; use
\({\mathcal D}_{\rm des}\) only for candidate selection.
\STATE Compute \(P(R,\bar Z_0)\) by solving \eqref{eq:p-budget} over
\({\mathcal B}(R,\bar Z_0)\).
\IF{\(\sigma_{\min}(\hat G_H)-\epsilon_{G,2}\leq0\)}
    \STATE discard \(q\); \textbf{continue}
\ENDIF
\IF{\eqref{eq:request-closure} fails}
    \STATE discard \(q\); \textbf{continue}
\ENDIF
\STATE Freeze \(q\), \(\hat O_H\), and \(\hat G_H\).
\FOR{each calibration episode \(i\in{\mathcal D}_{\rm cal}\)}
    \STATE Run the frozen update law \eqref{eq:controller} and store
    \((\Delta Y_k^{(i)},\Delta z_{k,0}^{(i)},\Delta U_k^{(i)})\).
    \STATE Set \(S_i\) by \eqref{eq:score}.
\ENDFOR
\STATE Compute the split-conformal quantile
\(\widehat q_{1-\beta_{\rm cal}}\).
\STATE Compute \(\bar w_\beta\) from \eqref{eq:hp-budget} and set
\(\Delta_{{\rm ISS},\beta}=\bar w_\beta/(1-\lambda)\).
\IF{\(\Delta_{{\rm ISS},\beta}\leq\Delta_{\rm req}\)}
    \FOR{deployment trials \(k=1,\ldots,K_{\max}\)}
        \STATE Form \(b_k\) by \eqref{eq:request}.
        \STATE Solve \eqref{eq:controller} for \(\Delta U_k\).
        \STATE Apply the admissible update and record
        \(p_k=b_k-\hat G_H\Delta U_k\).
    \ENDFOR
    \RETURN \(\textsc{Certified},q,\bar w_\beta,\Delta_{{\rm ISS},\beta}\).
\ENDIF
\ENDFOR
\RETURN \(\textsc{NotCertified}\).
\end{algorithmic}
\end{algorithm}

Algorithm~\ref{alg:iss-certification} is a rejection-capable certification
procedure. It returns \(\textsc{NotCertified}\) whenever the contraction
condition, channel margin, request closure, or certified band requirement
fails. The design split is used to select rows, gains, and admissible
updates; the calibration split is used only after the controller is frozen.
This ordering ensures the conformal score in Theorem~2 is computed under a
fixed, rather than adaptively updated, controller. During deployment the
algorithm solves the same constrained update problem that
appears in the proof and records the projection residual used in the ISS
budget.

\section{Trial-Axis Input-to-State Stability Certificates}

The central deterministic event states that the learned model, the true
selected channel, and the deployment protocol are close enough for the
trial-axis recursion to be closed.

\begin{assumption}[Certified deployment event]
For all certified deployment trials,
\begin{align}
\norm{\omega_k}_\infty&\leq \xi_0,&
\norm{\hat G_H-G_{\star,H}}_\infty&\leq\epsilon_G,
\nonumber\\
\norm{\hat O_H-O_{\star,H}}_\infty&\leq\epsilon_O,&
\norm{\Delta U_k}_\infty&\leq\bar U_\Delta,
\nonumber\\
\norm{\Delta z_{k,0}}_\infty&\leq\bar Z_0,&
\norm{p_k}_\infty&\leq\bar p.
\label{eq:deployment-event}
\end{align}
The learned selected channel also satisfies the lower-margin condition
\begin{equation}
\sigma_{\min}(\hat G_H)-\epsilon_{G,2}\geq \gamma_G>0 .
\label{eq:channel-margin}
\end{equation}
Here \(\epsilon_{G,2}\) denotes the corresponding spectral-norm perturbation
radius satisfying
\begin{equation}
\norm{\hat G_H-G_{\star,H}}_2\leq \epsilon_{G,2}.
\label{eq:spectral-channel-radius}
\end{equation}
\end{assumption}

\begin{lemma}[Trial-axis error recursion]
Under Assumption~1, the selected controlled error satisfies
\begin{equation}
\bar E_k=\Lambda\bar E_{k-1}+r_k
\label{eq:error-recursion}
\end{equation}
with
\begin{equation}
\norm{r_k}_\infty
\leq
\xi_0+\bar p+\epsilon_G\bar U_\Delta+\epsilon_O\bar Z_0 .
\label{eq:r-bound}
\end{equation}
\end{lemma}

\begin{proof}
Substitute the true selected channel \eqref{eq:true-channel} into the
definition of the next selected error. The desired increment identity
\eqref{eq:request} cancels the nominal selected channel and leaves
\(\Lambda\bar E_{k-1}\), the projection residual \(p_k\), the channel
perturbation
\((\hat G_H-G_{\star,H})\Delta U_k\), the reset/observability
perturbation
\((\hat O_H-O_{\star,H})\Delta z_{k,0}\), and the selected
residual \(\omega_k\). Applying \eqref{eq:deployment-event} gives
\eqref{eq:r-bound}.
\end{proof}

Theorem~1 shows that the learned-channel margin, projection budget, reset
mismatch, and residual calibration enter the learning recursion as ISS
inputs, rather than introducing a new comparison lemma.

\begin{theorem}[Deterministic practical ISS certificate]
Suppose Assumption~1 holds and define
\begin{equation}
\bar w=\xi_0+\bar p+\epsilon_G\bar U_\Delta+\epsilon_O\bar Z_0 .
\label{eq:wbar}
\end{equation}
Then the selected learning-error system is ISS along the trial axis:
\begin{equation}
\norm{\bar E_k}_\infty
\leq
\lambda^k\norm{\bar E_0}_\infty
+\frac{1-\lambda^k}{1-\lambda}\bar w .
\label{eq:iss-bound}
\end{equation}
Equivalently, \eqref{eq:piss-def} holds with
\(\beta(r,k)=\lambda^k r\), \(\gamma(s)=s/(1-\lambda)\), and \(c=0\), after
identifying \(s\) with the supremum of the certified perturbation budget.
\end{theorem}

\begin{proof}
From Lemma~2,
\(\norm{\bar E_k}_\infty\leq \lambda\norm{\bar E_{k-1}}_\infty+\bar w\).
Iterating this scalar comparison recursion gives \eqref{eq:iss-bound}. The
functions \(\beta\) and \(\gamma\) are of class \(\KL\) and \(\K\),
respectively, because \(0\leq\lambda<1\). Therefore the trial-axis system is
ISS on the certified deployment grid.
\end{proof}

Consequently, the main design problem becomes certificate construction:
computing \(\bar p\), verifying a channel margin, and calibrating \(\xi_0\),
rather than merely minimizing one-step prediction error.

\begin{corollary}[Componentwise ISS gains]
Let
\[
\bar w=\bar w_\omega+\bar w_p+\bar w_G+\bar w_O
\]
where
\(\bar w_\omega=\xi_0\), \(\bar w_p=\bar p\),
\(\bar w_G=\epsilon_G\bar U_\Delta\), and
\(\bar w_O=\epsilon_O\bar Z_0\). Then
\begin{equation}
\norm{\bar E_k}_\infty
\leq
\lambda^k\norm{\bar E_0}_\infty
+\sum_{q\in\{\omega,p,G,O\}}
\frac{1-\lambda^k}{1-\lambda}\bar w_q .
\label{eq:component-gains}
\end{equation}
Each residual source therefore shares the same scalar gain
\(\gamma_q(s)=s/(1-\lambda)\), though its numerical size is set by a
distinct certificate: residual calibration for \(\bar w_\omega\), projection
computation for \(\bar w_p\), channel identification for \(\bar w_G\), and
reset/observability certification for \(\bar w_O\).
\end{corollary}

\begin{remark}[Componentwise gain decomposition]
Equation~\eqref{eq:component-gains} exposes the individual ISS input
channels hidden inside a single aggregate ultimate-band budget. The
decomposition separates the part caused by unmodelled dynamics from the part
caused by constrained input reachability and finite-data uncertainty. This
is the certificate object used in the numerical section: disturbance, noise,
mismatch, and finite-sample effects are not generic perturbations, but
separate terms in the same gain estimate.
\end{remark}

\begin{corollary}[Finite-entry band]
For any \(\eta>0\), the selected error enters and remains in
\begin{equation}
{\mathcal X}_\eta=\left\{\bar E:\norm{\bar E}_\infty
\leq \frac{\bar w}{1-\lambda}+\eta\right\}
\label{eq:entry-band}
\end{equation}
after at most
\begin{equation}
k^\star(\eta)=
\left\lceil
\frac{\log(\eta/\norm{\bar E_0}_\infty)}{\log \lambda}
\right\rceil_+
\label{eq:entry-time}
\end{equation}
trials. The resulting finite-entry band is therefore the ultimate bound
induced by the ISS gain.
\end{corollary}

\begin{proposition}[ISS Lyapunov interpretation]
Let \(V(\bar E)=\norm{\bar E}_\infty\). Under Assumption~1,
\begin{equation}
V(\bar E_k)\leq \lambda V(\bar E_{k-1})+\norm{v_{k-1}}_{\rm ISS},
\label{eq:lyapunov-iss}
\end{equation}
where
\begin{align*}
\norm{v_{k-1}}_{\rm ISS}={}&
\xi_0+\norm{p_{k-1}}_\infty
+\epsilon_G\norm{\Delta U_{k-1}}_\infty \\
&+\epsilon_O\norm{\Delta z_{k-1,0}}_\infty .
\end{align*}
Consequently, for any \(\sigma\in(0,1-\lambda)\),
\begin{align*}
&V(\bar E_{k-1})\geq
\frac{\norm{v_{k-1}}_{\rm ISS}}{1-\lambda-\sigma}\\
&\quad\Longrightarrow\quad
V(\bar E_k)-V(\bar E_{k-1})
\leq -\sigma V(\bar E_{k-1}).
\end{align*}
Hence, \(V\) is an ISS Lyapunov function on the certified finite deployment
grid.
\end{proposition}

\begin{proof}
Inequality \eqref{eq:lyapunov-iss} is Lemma~2 before taking a uniform
supremum over the perturbation components. Subtracting \(V(\bar E_{k-1})\)
from both sides gives
\begin{align*}
V(\bar E_k)-V(\bar E_{k-1})
&\leq -(1-\lambda)V(\bar E_{k-1})\\
&\quad+\norm{v_{k-1}}_{\rm ISS}.
\end{align*}
The stated implication follows by upper bounding the input term by
\((1-\lambda-\sigma)V(\bar E_{k-1})\).
\end{proof}

\begin{remark}[Practical versus ideal ISS budgets]
In the ideal deterministic case the constant part of the budget can be made
zero by taking exact channels, no reset mismatch, no solver error, and
unconstrained full-row-rank updates. In the data-driven experiments the
budget is nonzero because finite samples, calibration scores, and
deployment disturbances are nonzero. Accordingly, the selected error is
driven to a data-dependent ultimate ball whose radius is computed before
deployment, which is what makes the result a practical, rather than
idealized, ISS statement.
\end{remark}

\begin{assumption}[Episode-level exchangeability]
Conditional on the fitted Koopman model and the frozen controller
\((\hat G_H,\hat O_H,S_H,H,\Lambda,\mu,{\mathcal U}_\Delta)\), the
calibration episodes and the deployment episode are exchangeable under the
declared resettable task protocol.
\end{assumption}

Assumption~2 is stated at the episode level, where each episode consists of
the reset state, disturbance realization, sensor noise, actuator mismatch,
and the closed-loop trajectory generated by the frozen controller. No
independence of individual time samples within an episode is required. If the reset
protocol, disturbance law, actuator limits, or reference task changes after
calibration, the split-conformal part of the certificate is no longer
claimed; only the deterministic ISS result applies with a separately
verified deployment-shift margin.

\begin{lemma}[Finite-horizon perturbation radii]
Suppose the identification split supplies one-step radii in a
submultiplicative induced norm,
\[
\norm{\widehat A-A_\star}\leq\varepsilon_A,\quad
\norm{\widehat B-B_\star}\leq\varepsilon_B,\quad
\norm{\widehat C-C_\star}\leq\varepsilon_C .
\]
For \(j\geq1\), define
\begin{equation}
D_A(j)=\varepsilon_A\sum_{\ell=0}^{j-1}
\norm{\widehat A}^{j-1-\ell}
\bigl(\norm{\widehat A}+\varepsilon_A\bigr)^\ell ,
\label{eq:DAj}
\end{equation}
and set \(D_A(0)=0\). Then the \(j\)-step observability block and Markov
block errors satisfy
\begin{align}
\delta_O(j)
={}&\norm{\widehat C}D_A(j)
+\varepsilon_C\bigl(\norm{\widehat A}+\varepsilon_A\bigr)^j,
\label{eq:deltaO}\\
\delta_G(j)
={}&\norm{\widehat C}\norm{\widehat A}^j\varepsilon_B
+\delta_O(j)\bigl(\norm{\widehat B}+\varepsilon_B\bigr).
\label{eq:deltaG}
\end{align}
Consequently, conservative stacked finite-horizon radii are
\begin{equation}
\varepsilon_{O,H,\beta}
=\sum_{j=0}^{H-1}\delta_O(j),\qquad
\varepsilon_{G,H,\beta}
=\sum_{j=0}^{H-1}(H-j)\delta_G(j).
\label{eq:finite-horizon-radii}
\end{equation}
For a row selector \(S_H\) with induced norm not exceeding one, the same
bounds apply to the selected stacked channel. Otherwise, the right-hand sides
in \eqref{eq:finite-horizon-radii} are multiplied by \(\norm{S_H}\).
The same construction in spectral norm gives a conservative
\(\epsilon_{G,2}\) for \eqref{eq:spectral-channel-radius}. These margins may
also be replaced by externally audited conservative bounds, but in either
case they are reported as separate certificate inputs rather than absorbed
into prediction loss.
\end{lemma}

\begin{proof}
The matrix-power identity gives
\[
\widehat A^j-A_\star^j
=\sum_{\ell=0}^{j-1}\widehat A^{j-1-\ell}
(\widehat A-A_\star)A_\star^\ell .
\]
Using \(\norm{A_\star}\leq\norm{\widehat A}+\varepsilon_A\) yields
\(\norm{\widehat A^j-A_\star^j}\leq D_A(j)\). Expanding
\(\widehat C\widehat A^j-C_\star A_\star^j\) gives
\eqref{eq:deltaO}. Expanding
\(\widehat C\widehat A^j\widehat B-C_\star A_\star^jB_\star\) and using
\(\norm{B_\star}\leq\norm{\widehat B}+\varepsilon_B\) gives
\eqref{eq:deltaG}. Summing the observability blocks and the repeated
Toeplitz Markov-block offsets over the finite horizon gives
\eqref{eq:finite-horizon-radii}.
\end{proof}

\begin{theorem}[Episode-calibrated finite-sample ISS budget]
Split resettable episodes into identification data \({\mathcal D}_\Theta\),
output-map data \({\mathcal D}_C\), calibration episodes
\({\mathcal D}_{\rm cal}\), and deployment episodes. Freeze
\(\hat G_H,\hat O_H,S_H,H,\Lambda,\mu\), and
\({\mathcal U}_\Delta\) before calibration. For each calibration episode
\(i\), define the complete-episode score
\begin{equation}
S_i=\max_{1\leq k\leq K_{\max}}
\norm{S_H(\Delta Y_k^{(i)}
-\hat O_H\Delta z_{k,0}^{(i)}
-\hat G_H\Delta U_k^{(i)})}_\infty .
\label{eq:score}
\end{equation}
Let \(n_{\rm cal}=|{\mathcal D}_{\rm cal}|\), let
\(S_{(1)}\leq\cdots\leq S_{(n_{\rm cal})}\) be the ordered scores, and set
\begin{equation}
r_{\rm cal}=\left\lceil (n_{\rm cal}+1)(1-\beta_{\rm cal})\right\rceil,
\qquad
\widehat q_{1-\beta_{\rm cal}}=S_{(r_{\rm cal})},
\label{eq:conformal-rank}
\end{equation}
with the usual convention that certification fails if
\(r_{\rm cal}>n_{\rm cal}\). If Assumption~2 holds, and if the
parameter/channel, shift, projection, and numerical events hold with risks
\(\beta_{\rm par}\), \(\beta_{\rm shift}\), \(\beta_{\rm proj}\), and
\(\beta_{\rm num}\), then with probability at least
\[
1-\beta,\qquad
\beta=\beta_{\rm par}+\beta_{\rm cal}+\beta_{\rm shift}
+\beta_{\rm proj}+\beta_{\rm num},
\]
Theorem~1 holds for \(k\leq K_{\max}\) with
\begin{align}
\bar w_\beta={}&
\widehat q_{1-\beta_{\rm cal}}
+\varepsilon_{G,H,\beta}\bar U_\Delta
+\varepsilon_{O,H,\beta}\bar Z_0 \nonumber\\
&+L_\Omega h_{{\rm shift},\beta}
+\Delta_w+\Delta_{{\rm num},\beta}+\bar p_\beta .
\label{eq:hp-budget}
\end{align}
\end{theorem}

\begin{proof}
For a frozen controller, the calibration scores and the deployment score are
exchangeable at the complete-episode level. The split-conformal rank
argument with \eqref{eq:conformal-rank} gives
\(S_{\rm dep}\leq\widehat q_{1-\beta_{\rm cal}}\) with probability at least
\(1-\beta_{\rm cal}\). This bounds the residual budget rather than
establishing closed-loop stability directly. The parameter/channel event
transfers the learned residual score to the true selected channel through
\(\varepsilon_{G,H,\beta}\bar U_\Delta\) and
\(\varepsilon_{O,H,\beta}\bar Z_0\). The shift, disturbance, numerical, and
projection events add the remaining terms in \eqref{eq:hp-budget}. A union
bound over the listed events gives probability at least \(1-\beta\). On the
intersection of these events, Assumption~1 holds with
\(\bar w\leq\bar w_\beta\), so Theorem~1 applies.
\end{proof}

\begin{remark}[Budget interpretation]
The finite-sample control budget is
\begin{align*}
\bar w_\beta={}&
\widehat q_{1-\beta_{\rm cal}}+\bar p_\beta
+\varepsilon_{G,H,\beta}\bar U_\Delta
+\varepsilon_{O,H,\beta}\bar Z_0\\
&+L_\Omega h_{{\rm shift},\beta}
+\Delta_{{\rm num},\beta}+\Delta_w .
\end{align*}
The terms are, respectively, the complete-episode residual score, the
projection/reachability residual, the learned-channel perturbation, the
reset/observability perturbation, the deployment-shift allowance, the
solver tolerance, and the declared disturbance inflation.
The certified ultimate radius is accordingly
\[
\Delta_{{\rm ISS},\beta}=\frac{\bar w_\beta}{1-\lambda}.
\]
The learning gain \(\lambda\) controls amplification, whereas the
data-driven protocol controls the size of the certified input budget.
\end{remark}

The finite-sample theorem is used as an audit trail rather than as a
single statistical post-processing step. First, the residual score
\(\widehat q_{1-\beta_{\rm cal}}\) measures complete-episode mismatch after
the controller has been frozen. This term is not allowed to absorb later
changes in the row selector, horizon, learning gain, actuator limits, or
reference task. Second, \(\bar p_\beta\) records the portion of the requested
output increment that the admissible update set cannot produce through the
learned channel, a control obstruction rather than a statistical estimation
error. Third, \(\varepsilon_{G,H,\beta}\bar U_\Delta\) and
\(\varepsilon_{O,H,\beta}\bar Z_0\) transfer identification and reset
uncertainty into the same trial-axis recursion. Finally, the shift,
solver, and declared disturbance terms state what has been added for
deployment beyond the calibrated resettable protocol.

This ordering allows each certification failure to be traced to a specific
cause. If the
residual score is too large, the learned lifted predictor or the noise
level is not compatible with the desired band. If the channel margin is not
positive, the data do not certify an inverse IOS gain for the selected
channel, regardless of prediction loss. If the projection residual is large,
the requested correction is outside the constrained reachable set. If the
numerical or shift allowance dominates, the theorem is signaling that the
implementation protocol, rather than the nominal learning law, sets the
ultimate radius. The numerical section follows this same order, so each
figure and table is tied to a term in \(\bar w_\beta\).

\begin{remark}[IOS interpretation of the selected channel]
The map \(\Delta U_k\mapsto \Delta\bar Y_k\) is a finite-horizon
input-output channel. The lower bound \eqref{eq:channel-margin} and the
projection budget \eqref{eq:p-budget} state when this IOS channel can
generate the output increment required by the ISS recursion.
\end{remark}

\FloatBarrier
\section{Numerical Checks}

The numerical studies evaluate the quantities required by the ISS bound,
rather than ranking learning controllers by tracking performance. Each
experiment checks one quantity from the theory: the complete-episode residual
score, the selected-channel margin, the projection budget, input/update
feasibility, and ultimate-band coverage after the certified entry index.

Table~\ref{tab:system-protocol} summarizes the certified task protocols used
in the audits. All examples belong to Assumption~0: the plant maps are
unknown to the controller, the reference is repeated over a finite horizon,
and the learning update is certified only for selected output rows and the
declared bounded deployment protocol.

\begin{table*}[t]
\caption{System and Protocol Summary for the Certificate Audits}
\label{tab:system-protocol}
\centering
\scriptsize
\setlength{\tabcolsep}{2pt}
\begin{tabular}{p{0.17\textwidth}p{0.23\textwidth}p{0.09\textwidth}p{0.09\textwidth}p{0.25\textwidth}p{0.12\textwidth}}
\toprule
System & Dynamics class and protocol & \((n,m,p)\) & \(N,K_{\max}\) &
Uncertainty and constraints & Certified object\\
\midrule
Main Duffing closed-loop audit & nonlinear discrete-time repetitive &
\((2,1,1)\) & \(90,26\) &
repeated reference; reset, disturbance, noise; input/update bounds &
selected error\\
Hard Duffing budget audit & nonlinear repetitive &
\((2,1,1)\) & \(70,35\) &
repeated reference; reset, disturbance, noise; deployment shift &
selected error\\
Cubic damping oscillator & nonlinear repetitive &
\((2,1,1)\) & \(70,35\) &
repeated reference; reset, disturbance, noise; deployment shift &
selected error\\
Nonlinear servo & nonlinear repetitive &
\((2,1,1)\) & \(70,35\) &
repeated reference; reset, disturbance, noise; deployment shift &
selected error\\
\bottomrule
\end{tabular}
\end{table*}

The main Duffing closed-loop audit is used for ultimate-band coverage,
whereas the hard Duffing budget audit is used for the finite-sample
certificate budget.

Before comparing numerical values, every run is reduced to the same
certificate audit. This avoids treating tracking curves, residual plots, and
ablation rows as unrelated empirical evidence. Table~\ref{tab:certificate-map}
lists the proof object being tested, where it appears in the main text, and
the decision it supports. The table is also the reason the reported
quantities emphasize budgets, margins, and certified entry times rather than
only final tracking errors.

\begin{table*}[t]
\caption{Main-Text Evidence for the Input-to-State Stability Certificate}
\label{tab:certificate-map}
\centering
\scriptsize
\begin{tabular}{p{0.18\textwidth}p{0.34\textwidth}p{0.36\textwidth}}
\toprule
Certificate object & Main-text evidence & Certification role\\
\midrule
\(\widehat q_{1-\beta_{\rm cal}}\) and \(\bar w\) &
Figs.~\ref{fig:budget-stack}--\ref{fig:residual-band},
Tables~\ref{tab:finite-budget} and~\ref{tab:gain-sweep} &
Calibrates the ISS input budget\\
\(\Delta_{\rm ISS}=\bar w/(1-\lambda)\) &
Fig.~\ref{fig:control-traces}, Tables~\ref{tab:closed-loop}
and~\ref{tab:gain-sweep} &
Checks ultimate-band coverage\\
\(\gamma_G\) and channel informativeness &
Fig.~\ref{fig:weak-channel}, Tables~\ref{tab:weak-channel-audit}
and~\ref{tab:ablation-decisions} &
Rejects weak selected channels\\
\(\bar p\) and projection closure &
Fig.~\ref{fig:rejection-cascade}, Tables~\ref{tab:candidate-log}
and~\ref{tab:ablation-decisions} &
Tests reachability of requested updates\\
Input/update feasibility &
Fig.~\ref{fig:control-traces}, Table~\ref{tab:closed-loop} &
Checks admissible learning actions\\
\bottomrule
\end{tabular}
\end{table*}

The finite-sample audit is generated from the retained experiment summaries
by the checker scripts in the code archive. Fig.~\ref{fig:budget-stack}
shows the additive terms in \(\bar w_\beta\), while
Fig.~\ref{fig:conformal-coverage} reports the episode-level conformal
coverage check used for the residual term. Here \(n_{\rm cal}\) denotes the
number of independently generated complete calibration episodes. Each
episode consists of the reset state, disturbance realization, sensor noise,
actuator mismatch, and the full closed-loop trajectory generated by the
frozen controller; it is not the number of within-episode time samples.
These calibration episodes are Monte Carlo resettable simulation episodes
used to audit the finite-sample certificate. Physical deployment with fewer
episodes would yield a larger conformal quantile and hence a more
conservative certified band.

\begin{figure}[t]
\centering
\includegraphics[width=\linewidth]{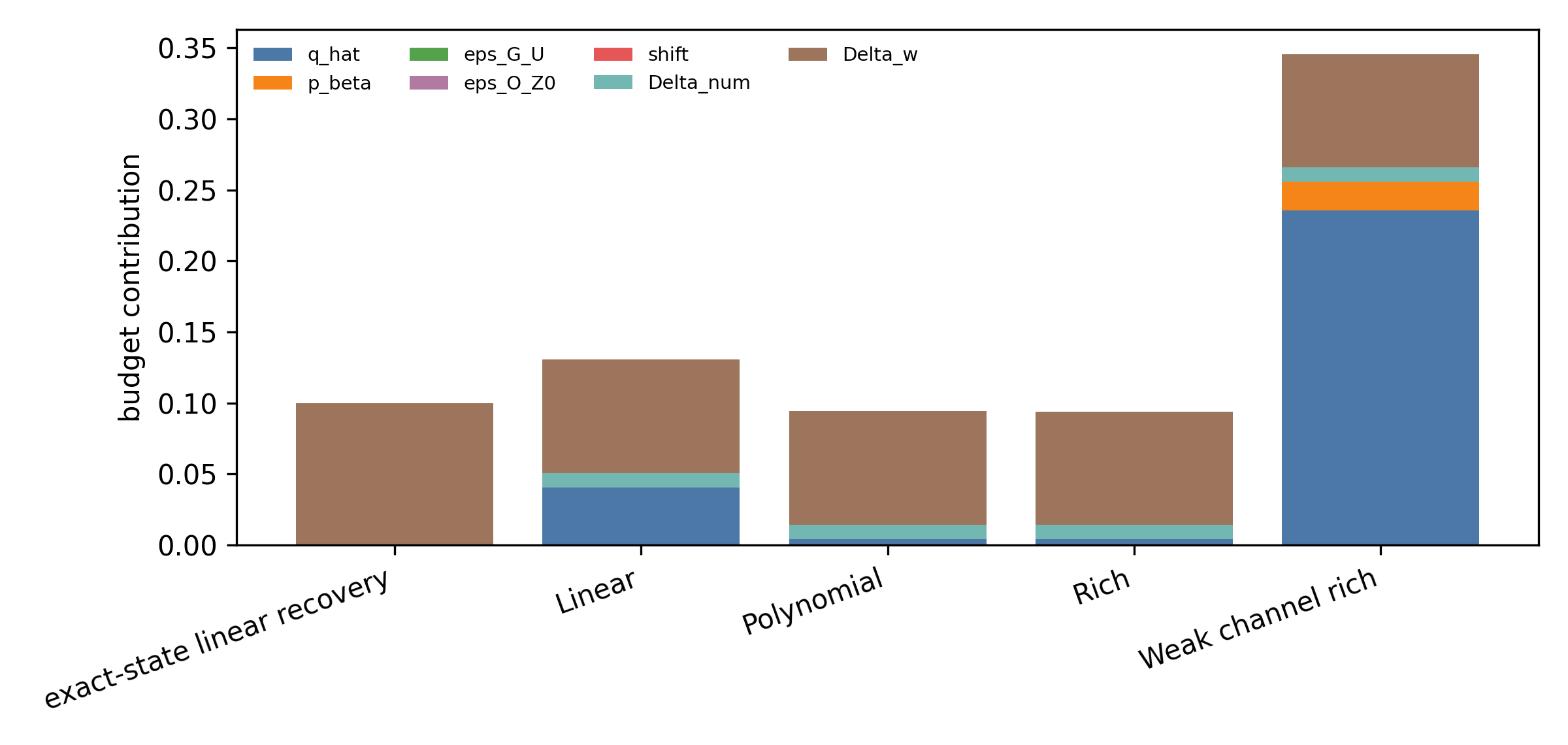}
\caption{Finite-sample ISS budget decomposition. Each bar is the sum of the
episode residual score, projection term, channel and reset margins, shift
allowance, numerical tolerance, and declared disturbance inflation.}
\label{fig:budget-stack}
\end{figure}

\begin{figure}[t]
\centering
\includegraphics[width=\linewidth]{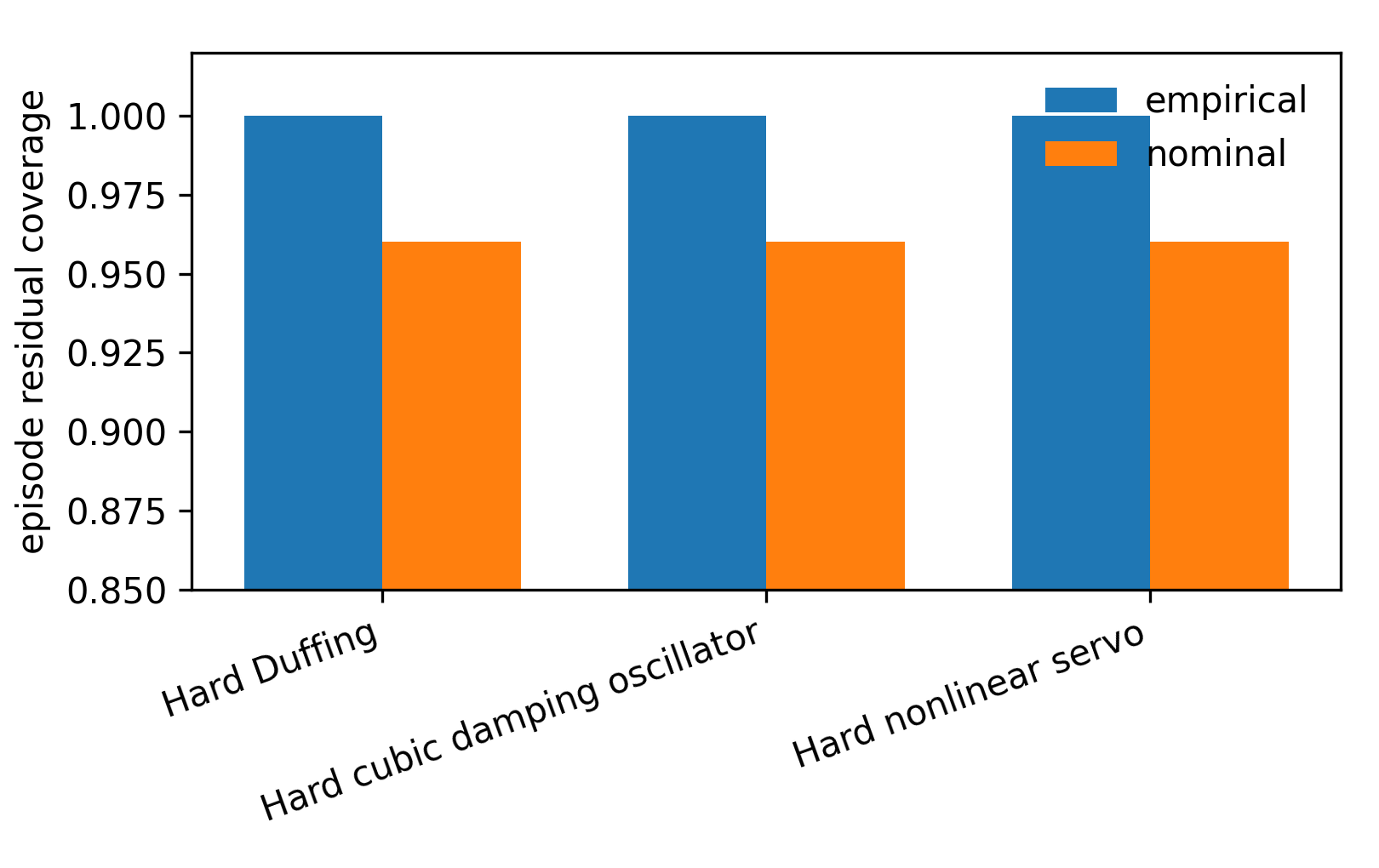}
\caption{Episode-level split-conformal residual audit. The controller is
frozen before calibration, and the rank
\(r_{\rm cal}=\lceil(n_{\rm cal}+1)(1-\beta_{\rm cal})\rceil\) determines
the reported residual quantile.}
\label{fig:conformal-coverage}
\end{figure}

\begin{table*}[t]
\caption{Finite-Sample ISS Budget Audit from the Generated CSV}
\label{tab:finite-budget}
\centering
\scriptsize
\begin{tabular}{lcccccccc}
\toprule
Case & \(n_{\rm cal}\) episodes & \(r_{\rm cal}\) & \(\widehat q\) &
\(p_\beta\) & \(G/O\) & shift/num/\(\Delta_w\) &
\(\bar w_\beta\) & Band\\
\midrule
Hard Duffing valid-channel & 120000 & 115201 & 0.330 & 0.000 &
0.106 & 0.152 & 0.588 & 1.110\\
Hard cubic damping valid-channel & 120000 & 115201 & 0.209 & 0.000 &
0.069 & 0.134 & 0.412 & 0.778\\
Hard nonlinear servo valid-channel & 120000 & 115201 & 0.111 & 0.000 &
0.122 & 0.168 & 0.402 & 0.758\\
\bottomrule
\end{tabular}
\end{table*}

Negative cases are interpreted accordingly: a row that violates a hypothesis
before the ISS recursion is invoked is treated as a certificate rejection
rather than a failed simulation. In particular, a weak selected channel may still yield accurate held-out
predictions, but it cannot support the certified projection step used to
convert the requested output correction into an admissible input update.
Conversely, conservatively inflating \(\bar w\) certifies the same
controller under a larger admissible input budget, enlarging the ultimate
band through the ISS gain.

The first experiment is a two-state Duffing-type nonlinear repetitive plant
identified only through the learned Koopman model. The selected output is
the tracked second state; the deployment includes reset variation,
bounded disturbance, input limits, and a finite learning horizon. The
closed-loop trace in Fig.~\ref{fig:control-traces} shows the controlled
state, output, input, and update magnitude. Table~\ref{tab:closed-loop}
lists the quantities used by the ISS certificate.

\begin{figure}[t]
\centering
\includegraphics[width=\linewidth]{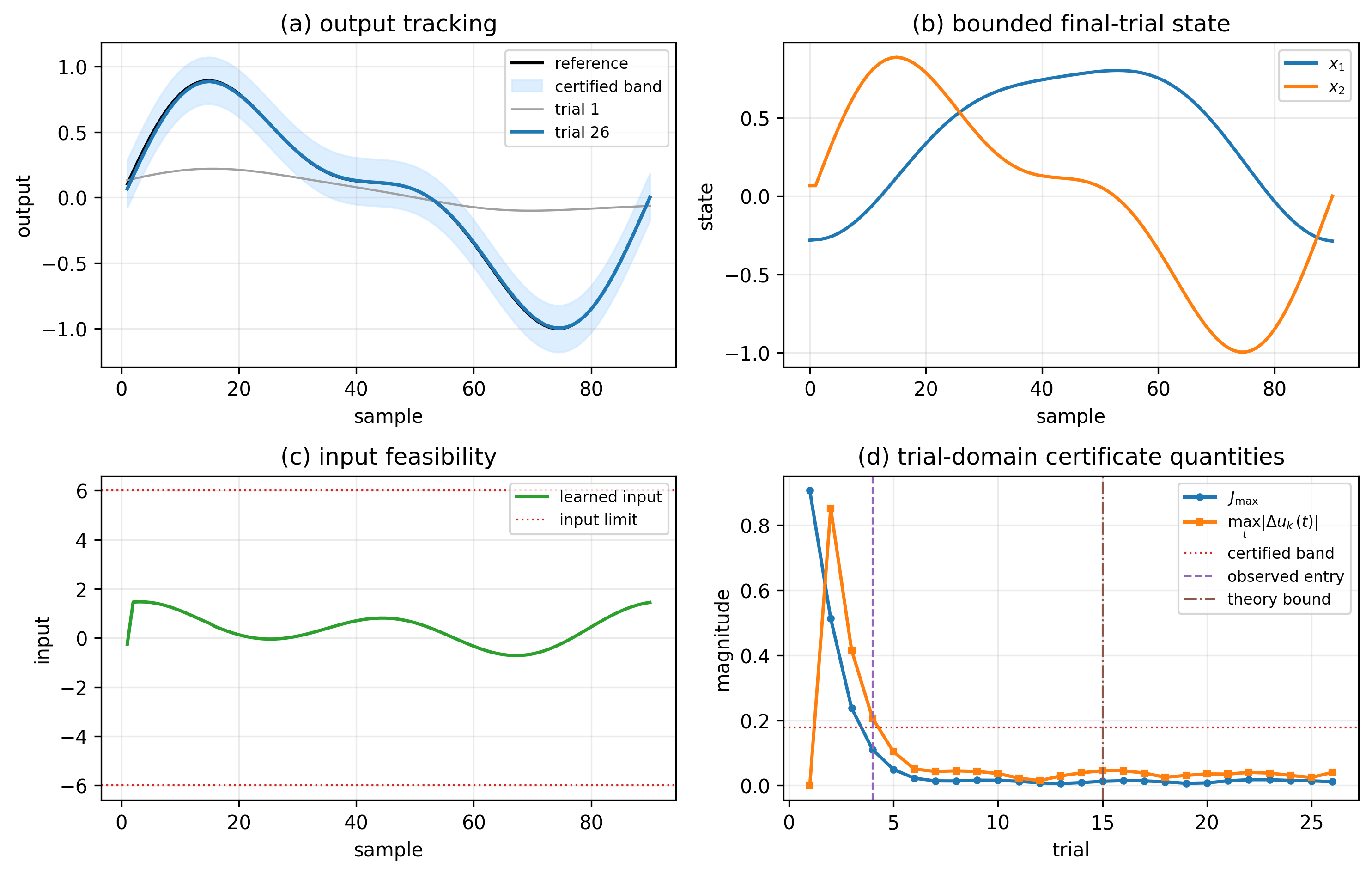}
\caption{Closed-loop state, output, input, and input-update traces for the
residual-certified Koopman learning controller. The certified ultimate band
is the ISS gain applied to the calibrated perturbation budget.}
\label{fig:control-traces}
\end{figure}

\begin{table}[t]
\caption{Closed-Loop ISS Certificate Quantities in the Main Duffing Run}
\label{tab:closed-loop}
\centering
\scriptsize
\begin{tabular}{lcc}
\toprule
Quantity & Value & Role\\
\midrule
Certified band \(\Delta_{\rm ISS}\) & 0.1777 & ultimate ISS bound\\
Observed entry trial & 4 & empirical finite-entry index\\
Theoretical entry bound & 15 & certified \(k^\star\)\\
Final \(J_{\max}\) & 0.0113 & terminal selected error\\
Post-entry max \(J_{\max}\) & 0.1102 & inside certified band\\
Max \(|u|\) & 1.6002 & input feasibility\\
Max \(|\Delta u|\) & 0.8520 & update feasibility\\
\bottomrule
\end{tabular}
\end{table}

The second experiment varies residual and uncertainty levels. In the ISS
view, this sweep serves as a gain test rather than a conventional
sensitivity analysis. Increasing
measurement noise, disturbance, or model mismatch increases the certified
budget \(\bar w\), and Theorem~1 predicts a monotone enlargement of the
ultimate band. Table~\ref{tab:gain-sweep} reports representative rows from
the generated sweep. Fig.~\ref{fig:residual-band} gives the corresponding
residual-to-band validation.

\begin{figure}[t]
\centering
\includegraphics[width=\linewidth]{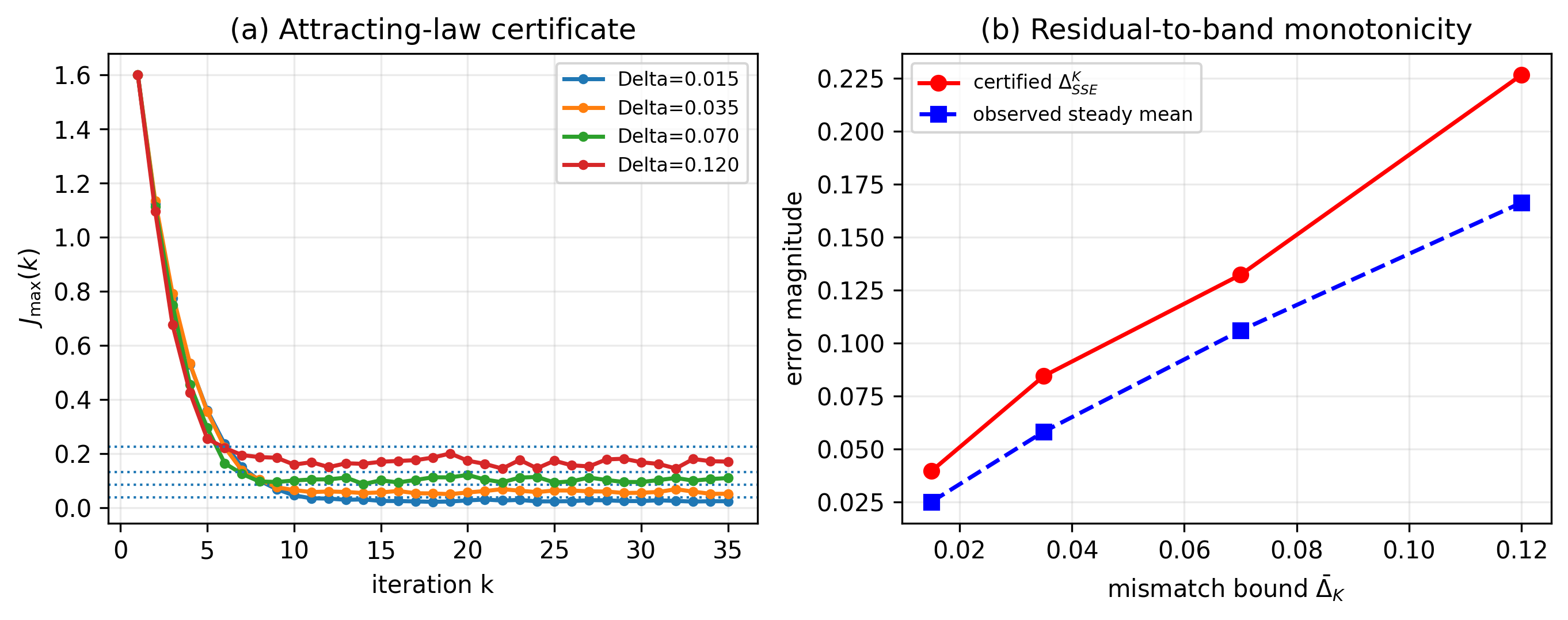}
\caption{Residual-to-band validation. The certified band grows with the
aggregate perturbation budget, which is the empirical counterpart of the
ISS gain \(\gamma(s)=s/(1-\lambda)\).}
\label{fig:residual-band}
\end{figure}

\begin{table*}[t]
\caption{Representative ISS-Gain Sweep Rows from \texttt{iss\_gain\_sweep.csv}}
\label{tab:gain-sweep}
\centering
\scriptsize
\begin{tabular}{llccccccccc}
\toprule
Family & Level & \(\widehat q\) & \(p_\beta\) & channel & reset &
other & \(\bar w\) & Band & \(\rho\) & \(k^\star\)\\
\midrule
disturbance & 0.005 & 0.001 & 0.000 & 0.000 & 0.000 &
0.090 & 0.091 & 0.172 & 0.10 & 15\\
disturbance & 0.090 & 0.026 & 0.000 & 0.000 & 0.000 &
0.260 & 0.286 & 0.540 & 0.10 & 13\\
noise & 0.00 & 0.004 & 0.000 & 0.000 & 0.000 &
0.110 & 0.114 & 0.216 & 0.10 & 14\\
noise & 0.10 & 0.038 & 0.000 & 0.000 & 0.000 &
0.210 & 0.248 & 0.469 & 0.10 & 13\\
mismatch & 0.00 & 0.004 & 0.000 & 0.000 & 0.000 &
0.110 & 0.114 & 0.215 & 0.10 & 15\\
mismatch & 0.22 & 0.005 & 0.000 & 0.000 & 0.000 &
0.330 & 0.335 & 0.631 & 0.10 & 13\\
\bottomrule
\end{tabular}
\end{table*}

These rows are reported because they correspond directly to
\eqref{eq:wbar} and \eqref{eq:iss-bound}: changing disturbance, noise, or
mismatch changes a certified input budget, and the ultimate radius changes
through the same scalar gain \(1/(1-\lambda)\). The table therefore reports
\(\bar w\), the band, and \(k^\star\), rather than only terminal tracking
error.

The third experiment tests control-channel certification.
Fig.~\ref{fig:weak-channel} compares learned models that differ in the
margin of the selected finite-horizon channel. The weak-channel case is
rejected because the selected channel fails the margin and projection
conditions required by the ISS certificate.

\begin{figure}[t]
\centering
\includegraphics[width=\linewidth]{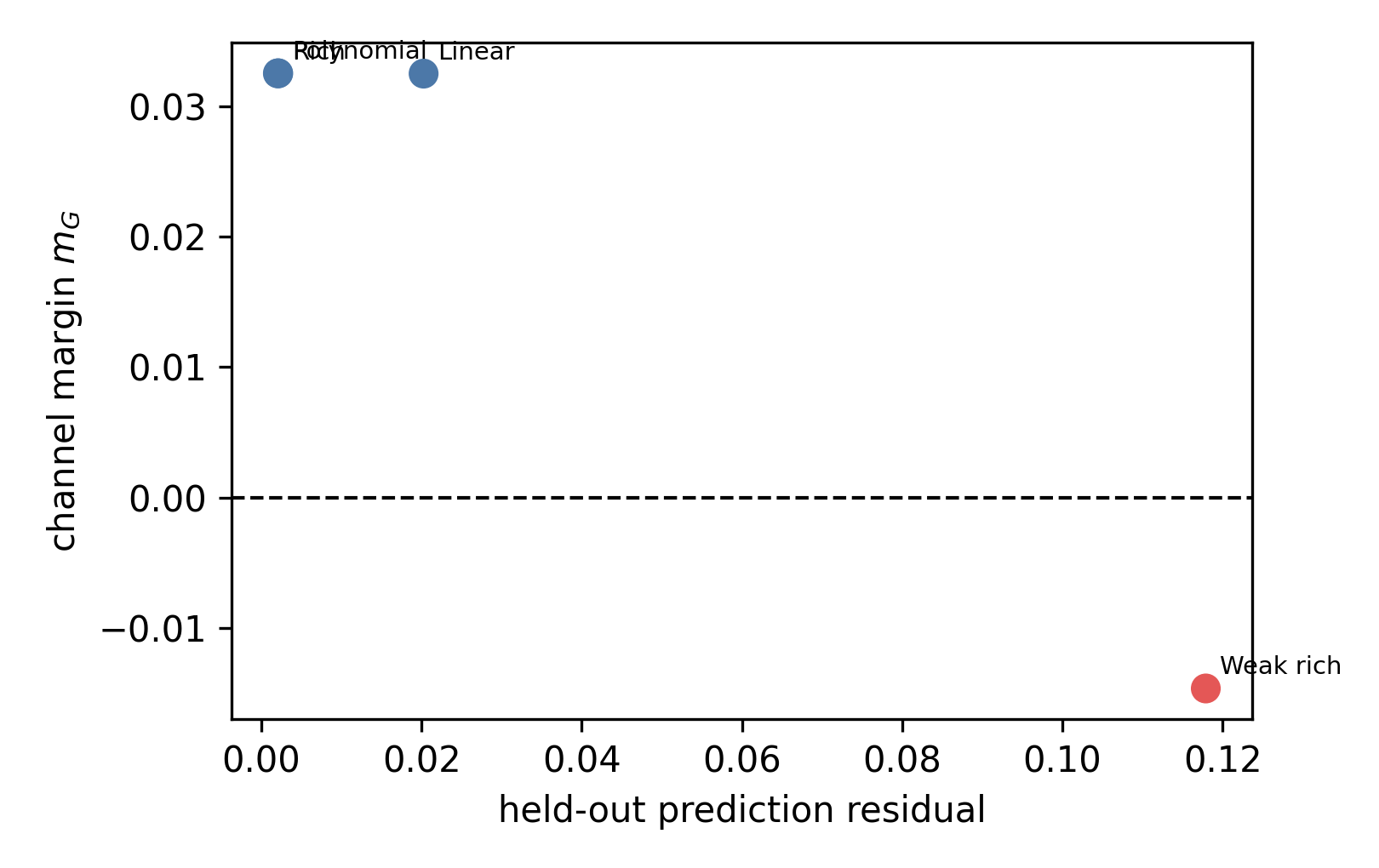}
\caption{Selected-channel certification ablation. The weak-channel case is
rejected because the selected input-output channel does not pass the
lower-margin test, even though residual size is also reported.}
\label{fig:weak-channel}
\end{figure}

\begin{table}[t]
\caption{Weak-Channel Audit}
\label{tab:weak-channel-audit}
\centering
\scriptsize
\begin{tabular}{lcccc}
\toprule
Case & residual & \(\sigma_{\min}\) & \(m_G\) & Cert.\\
\midrule
Linear dictionary & 0.020 & 0.057 & 0.032 & yes\\
Polynomial dictionary & 0.002 & 0.058 & 0.033 & yes\\
Rich dictionary & 0.002 & 0.058 & 0.033 & yes\\
Weak-channel rich & 0.118 & 0.010 & -0.015 & no\\
\bottomrule
\end{tabular}
\end{table}

Candidate rejection is logged before the deterministic ISS theorem is
invoked. Fig.~\ref{fig:rejection-cascade} and
Table~\ref{tab:candidate-log} show that the implemented protocol contains
both margin-failure and request-closure-failure cases, plus certified
accepted cases.

\begin{figure}[t]
\centering
\includegraphics[width=\linewidth]{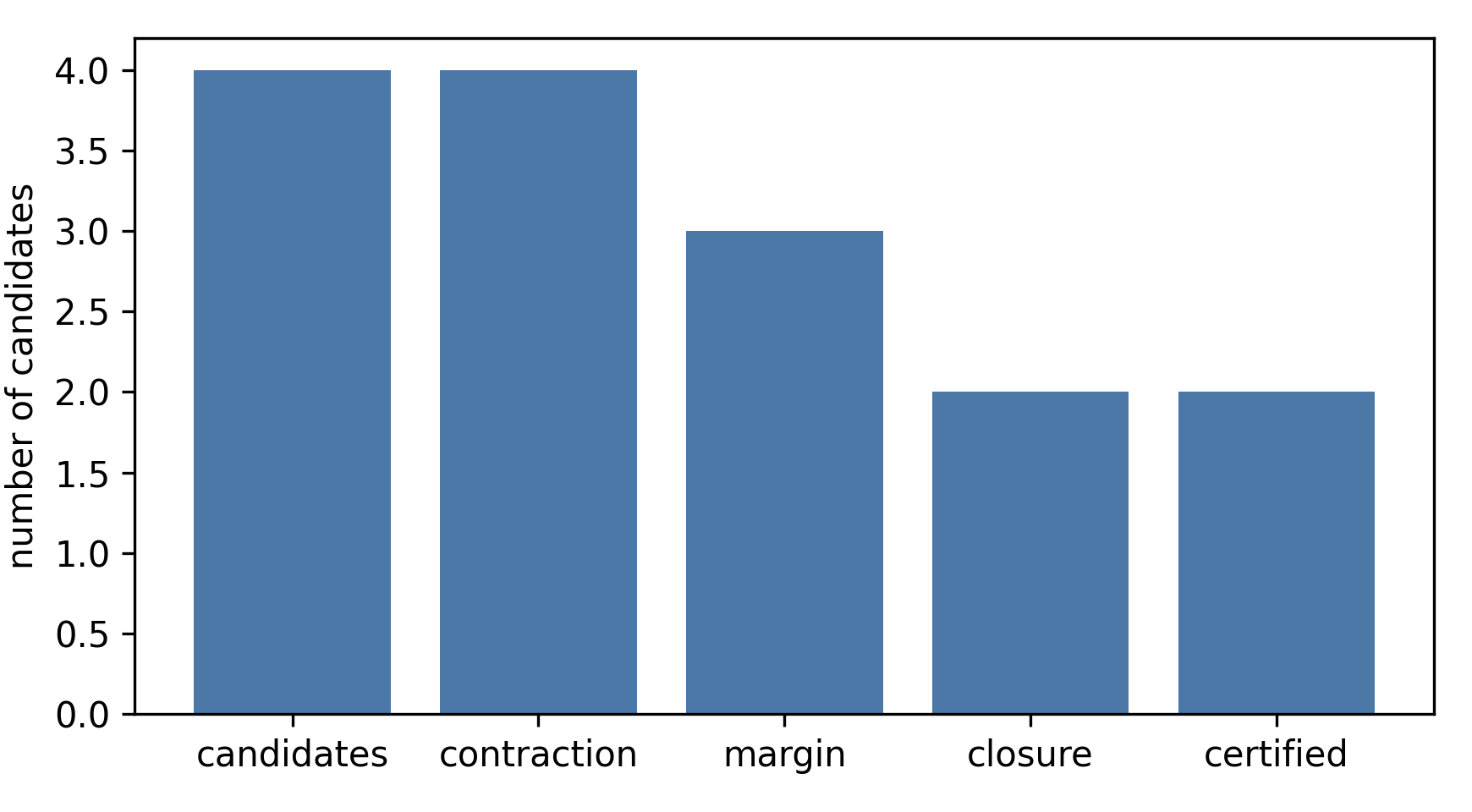}
\caption{Candidate rejection cascade. The certificate rejects candidates at
the first failed gate: contraction, channel margin, request closure, or
ultimate-band requirement.}
\label{fig:rejection-cascade}
\end{figure}

\begin{table}[t]
\caption{Candidate Rejection Log}
\label{tab:candidate-log}
\centering
\scriptsize
\begin{tabular}{lcccc}
\toprule
ID & \(m_G\) & \(P_R\) & closure/R & Decision\\
\midrule
q1 & -0.015 & 0.020 & 0.61 & margin fail\\
q2 & 0.033 & 0.310 & 1.07 & request-closure fail\\
q3 & 0.033 & 0.030 & 0.72 & certified\\
q4 & 0.033 & 0.000 & 0.58 & certified\\
\bottomrule
\end{tabular}
\end{table}

The baseline ablation is best understood as a set of certificate decisions,
not as a ranking of controller names. Table~\ref{tab:ablation-decisions}
keeps the same closed-loop task and asks which ingredients close the ISS
argument. The nominal rich Koopman controller has a useful learned predictor,
but the reported certificate sets the residual input to zero; the observed
trajectory therefore violates the claimed zero-residual band even though the
same predictor becomes certifiable once the calibrated residual budget is
inserted. The inflated-budget row uses the same controller again, showing the
expected ISS tradeoff: the certified band expands when the input budget is
made more conservative, while the trajectory is still covered. The
weak-channel row is central to the paper's argument: although its residual
budget is large, the decisive failure is the channel margin, which blocks the
projection step required by Assumption~1. External LQR-ILC, neural ILC, and Koopman-MPC baselines are useful
tracking-performance comparators, but they do not provide input-to-state
stability certificates unless they report the same fixed channel, projection
residual, and perturbation-budget quantities.

\begin{table}[t]
\caption{Ablation Rows Read as ISS Decisions}
\label{tab:ablation-decisions}
\centering
\scriptsize
\begin{tabular}{p{0.38\linewidth}ccc}
\toprule
Case & Budget & Band & ISS cert.\\
\midrule
Linear proxy baseline & 0.1305 & 0.2463 & yes\\
Nominal rich, zero residual & 0.0000 & 0.0001 & no\\
Residual-certified rich & 0.0941 & 0.1777 & yes\\
Inflated-budget rich & 0.1647 & 0.3109 & yes\\
Weak-channel rich & 0.3458 & 0.6525 & no\\
\bottomrule
\end{tabular}
\end{table}

The numerical evidence is deliberately scoped. It verifies the algebraic
objects used by the proofs: residual score, channel lower bound, projection
budget, input/update feasibility, and ultimate-band coverage after the
certified entry index. Additional benchmark rows and checker outputs are
included in the data and code archive.

\FloatBarrier
\section{Discussion}

The guarantee is finite-horizon and task-protocol dependent. The row
selector \(S_H\), horizon \(H\), learning matrix \(\Lambda\), admissible
update set, and calibration score are fixed before deployment. Within this
certified setting, the learning-error state decays geometrically from its
initial value, and all uncertainty sources enter through the computable
budget \(\bar w\). In this sense, the result constitutes a practical ISS
certificate rather than a universal ISS theorem for arbitrary learned
nonlinear controllers.

The result also separates three quantities that are often conflated in
data-driven control: model fit, finite-horizon IOS channel informativeness,
and trial-axis ISS of the learning recursion. The weak-channel case is
therefore rejected through non-certifiability of the selected channel, not
through a closed-loop instability claim.

The finite-sample statement has the same boundary. Split conformal
calibration supplies a complete-episode residual budget under an
exchangeable resettable protocol and a frozen controller. It does not by
itself assert stability, and it does not apply after the reset law,
disturbance law, actuator limits, or reference task changes. Stability is
obtained by applying the deterministic ISS theorem on the event where the
calibrated budget, channel margin, projection budget, and deployment-shift
margin hold.

\section{Conclusion}

We have formulated Koopman learning control of unknown nonlinear repetitive
systems as a practical input-to-state stability certification problem along
the learning-trial axis. The deterministic result gives an explicit
ultimate-band estimate, while the finite-sample implementation constructs the
corresponding perturbation budget from complete-episode residuals, channel
and reset margins, projection residuals, deployment-shift allowances, and
numerical tolerances. The construction shows that Koopman learning control
can be certified only when model residuals, projection residuals, and
selected-channel margins are reported as separate stability quantities.
Numerical checks support the predicted residual-to-band scaling and the
rejection of weak selected channels by non-certifiability, without claiming
that such predictors are inaccurate or that all closed-loop trajectories are
unstable.

\section*{Data and code availability}

The data and code supporting the numerical results are available in the
public GitHub repository associated with this submission:
\href{https://github.com/marcowus/projection-residual-iss-koopman-learning-control}{\nolinkurl{https://github.com/marcowus/projection-residual-iss-koopman-learning-control}}.
The archive contains the reproducibility package used to generate the
figures, tables, verification CSV files, and checker outputs
reported in this article.

\section*{Funding}

This work was supported by the Key Research and Development Program of
Shaanxi Province, Four-Chain Integration Major Project, ``Research and
development of key technologies and equipment for wall-working construction
robots based on digital twins'' [grant number 2024PT-ZCK-77; project period:
2024.10.1--2026.9.30]. The funding source had no role in the study design,
data collection and analysis, manuscript preparation, or decision to submit
the article for publication.

\section*{Declaration of competing interest}

The authors declare that they have no known competing financial interests or
personal relationships that could have appeared to influence the work
reported in this article.

\section*{CRediT authorship contribution statement}

Yue Wu: Conceptualization, Methodology, Software, Validation, Formal
analysis, Writing--original draft. Ye Cao: Investigation, Validation,
Visualization, Writing--review and editing. Jianfu Cao: Supervision, Funding
acquisition, Project administration, Writing--review and editing.

\section*{Declaration of generative AI and AI-assisted technologies in the manuscript preparation process}

During the preparation of this work, the authors used ChatGPT for language
editing and translation support. The authors reviewed, revised, and verified
all mathematical statements, experimental descriptions, and conclusions, and
take full responsibility for the final manuscript.


\begingroup\scriptsize
\begin{thebibliography}{99}
\setlength{\itemsep}{0pt}

\bibitem{sontag1989}
E. D. Sontag, ``Smooth stabilization implies coprime factorization,''
\emph{IEEE Transactions on Automatic Control}, vol. 34, no. 4, pp. 435--443,
1989, doi: 10.1109/9.28018.

\bibitem{sontagwang1995}
E. D. Sontag and Y. Wang, ``On characterizations of the input-to-state
stability property,'' \emph{Systems \& Control Letters}, vol. 24, no. 5,
pp. 351--359, 1995, doi: 10.1016/0167-6911(94)00050-6.

\bibitem{jiangteelpraly1994}
Z.-P. Jiang, A. R. Teel, and L. Praly, ``Small-gain theorem for ISS systems
and applications,'' \emph{Mathematics of Control, Signals and Systems},
vol. 7, no. 2, pp. 95--120, 1994, doi: 10.1007/BF01211469.

\bibitem{jiangwang2001}
Z.-P. Jiang and Y. Wang, ``Input-to-state stability for discrete-time
nonlinear systems,'' \emph{Automatica}, vol. 37, no. 6, pp. 857--869, 2001,
doi: 10.1016/S0005-1098(01)00028-0.

\bibitem{angelisontagwang2000}
D. Angeli, E. D. Sontag, and Y. Wang, ``A characterization of integral
input-to-state stability,'' \emph{IEEE Transactions on Automatic Control},
vol. 45, no. 6, pp. 1082--1097, 2000, doi: 10.1109/9.863594.

\bibitem{angeli2002}
D. Angeli, ``A Lyapunov approach to incremental stability properties,''
\emph{IEEE Transactions on Automatic Control}, vol. 47, no. 3, pp. 410--421,
2002, doi: 10.1109/9.989067.

\bibitem{dashkovskiy2010}
S. N. Dashkovskiy, B. S. R\"uffer, and F. R. Wirth, ``Small gain theorems
for large scale systems and construction of ISS Lyapunov functions,''
\emph{SIAM Journal on Control and Optimization}, vol. 48, no. 6,
pp. 4089--4118, 2010, doi: 10.1137/090746483.

\bibitem{mironchenko2018}
A. Mironchenko and F. Wirth, ``Characterizations of input-to-state stability
for infinite-dimensional systems,'' \emph{IEEE Transactions on Automatic
Control}, vol. 63, no. 6, pp. 1692--1707, 2018, doi:
10.1109/TAC.2017.2756341.

\bibitem{arimoto1984}
S. Arimoto, S. Kawamura, and F. Miyazaki, ``Bettering operation of robots by
learning,'' \emph{Journal of Robotic Systems}, vol. 1, no. 2, pp. 123--140,
1984, doi: 10.1002/rob.4620010203.

\bibitem{moore1993}
K. L. Moore, \emph{Iterative Learning Control for Deterministic Systems}.
London, U.K.: Springer-Verlag, 1993, doi: 10.1007/978-1-4471-1912-8.

\bibitem{amann1996}
N. Amann, D. H. Owens, and E. Rogers, ``Iterative learning control using
optimal feedback and feedforward actions,'' \emph{International Journal of
Control}, vol. 65, no. 2, pp. 277--293, 1996, doi:
10.1080/00207179608921697.

\bibitem{gao2001robust}
F. Gao, Y. Yang, and C. Shao, ``Robust iterative learning control with
applications to injection molding process,'' \emph{Chemical Engineering
Science}, vol. 56, no. 24, pp. 7025--7034, 2001, doi:
10.1016/S0009-2509(01)00339-6.

\bibitem{bristow2006survey}
D. A. Bristow, M. Tharayil, and A. G. Alleyne, ``A survey of iterative
learning control,'' \emph{IEEE Control Systems Magazine}, vol. 26, no. 3,
pp. 96--114, 2006, doi: 10.1109/MCS.2006.1636313.

\bibitem{koopman1931}
B. O. Koopman, ``Hamiltonian systems and transformation in Hilbert space,''
\emph{Proceedings of the National Academy of Sciences of the United States
of America}, vol. 17, no. 5, pp. 315--318, 1931, doi:
10.1073/pnas.17.5.315.

\bibitem{rowley2009}
C. W. Rowley, I. Mezi\'c, S. Bagheri, P. Schlatter, and D. S. Henningson,
``Spectral analysis of nonlinear flows,'' \emph{Journal of Fluid
Mechanics}, vol. 641, pp. 115--127, 2009, doi:
10.1017/S0022112009992059.

\bibitem{williams2015}
M. O. Williams, I. G. Kevrekidis, and C. W. Rowley, ``A data-driven
approximation of the Koopman operator: Extending dynamic mode
decomposition,'' \emph{Journal of Nonlinear Science}, vol. 25, no. 6,
pp. 1307--1346, 2015, doi: 10.1007/s00332-015-9258-5.

\bibitem{brunton2016}
S. L. Brunton, B. W. Brunton, J. L. Proctor, and J. N. Kutz, ``Koopman
invariant subspaces and finite linear representations of nonlinear
dynamical systems for control,'' \emph{PLOS ONE}, vol. 11, no. 2, e0150171,
2016, doi: 10.1371/journal.pone.0150171.

\bibitem{proctor2016}
J. L. Proctor, S. L. Brunton, and J. N. Kutz, ``Dynamic mode decomposition
with control,'' \emph{SIAM Journal on Applied Dynamical Systems}, vol. 15,
no. 1, pp. 142--161, 2016, doi: 10.1137/15M1013857.

\bibitem{korda2018}
M. Korda and I. Mezi\'c, ``Linear predictors for nonlinear dynamical
systems: Koopman operator meets model predictive control,''
\emph{Automatica}, vol. 93, pp. 149--160, 2018, doi:
10.1016/j.automatica.2018.03.046.

\bibitem{mauroy2020}
A. Mauroy and J. Goncalves, ``Koopman-based lifting techniques for nonlinear
systems identification,'' \emph{IEEE Transactions on Automatic Control},
vol. 65, no. 6, pp. 2550--2565, 2020, doi: 10.1109/TAC.2019.2941433.

\bibitem{nueske2023}
F. N\"uske, S. Peitz, F. Philipp, M. Schaller, and K. Worthmann, ``Finite
data error bounds for Koopman-based prediction and control,''
\emph{Journal of Nonlinear Science}, vol. 33, no. 1, 2023, doi:
10.1007/s00332-022-09862-1.

\bibitem{dean2020}
S. Dean, H. Mania, N. Matni, B. Recht, and S. Tu, ``On the sample complexity
of the linear quadratic regulator,'' \emph{Foundations of Computational
Mathematics}, vol. 20, no. 4, pp. 633--679, 2020, doi:
10.1007/s10208-019-09426-y.

\bibitem{depersis2020}
C. De Persis and P. Tesi, ``Formulas for data-driven control:
Stabilization, optimality, and robustness,'' \emph{IEEE Transactions on
Automatic Control}, vol. 65, no. 3, pp. 909--924, 2020, doi:
10.1109/TAC.2019.2959924.

\bibitem{berberich2021}
J. Berberich, J. K{\"o}hler, M. A. M{\"u}ller, and F. Allg{\"o}wer,
``Data-driven model predictive control with stability and robustness
guarantees,'' \emph{IEEE Transactions on Automatic Control}, vol. 66,
no. 4, pp. 1702--1717, 2021, doi: 10.1109/TAC.2020.3000182.

\bibitem{vovk2005}
V. Vovk, A. Gammerman, and G. Shafer, \emph{Algorithmic Learning in a
Random World}. New York, NY, USA: Springer, 2005, doi: 10.1007/b106715.

\bibitem{lei2018}
J. Lei, M. G'Sell, A. Rinaldo, R. Tibshirani, and L. Wasserman,
``Distribution-free predictive inference for regression,'' \emph{Journal
of the American Statistical Association}, vol. 113, no. 523, pp. 1094--1111,
2018, doi: 10.1080/01621459.2017.1307116.

\end{thebibliography}
\endgroup
\end{document}